\documentclass[12pt,a4paper]{article}
\usepackage[ansinew]{inputenc}
\usepackage[T1]{fontenc}
\usepackage[english]{babel}
\usepackage{color}
\usepackage{times}
\usepackage{multicol}
\usepackage{graphicx}
\usepackage{geometry}
\usepackage{amssymb, amsmath, amsthm, amstext}
\newtheorem{Thm}{\bf Theorem}[section]
\newtheorem{Ass}[Thm]{\bf Assumption}
\newtheorem{Def}[Thm]{\bf Definition}
\newtheorem{Cor}[Thm]{\bf Corollary}
\newtheorem{Lem}[Thm]{\bf Lemma}
\newtheorem{Prop}[Thm]{\bf Proposition}
\theoremstyle{remark}
\newtheorem{Rem}[Thm]{\bf Remark}

\numberwithin{equation}{section}

\newcommand{\be}{\begin{equation}}
\newcommand{\ee}{\end{equation}}
\newcommand{\ba}{\begin{aligned}}
\newcommand{\ea}{\end{aligned}}
\newcommand{\R}{\mathbb{R}}

\newcommand{\C}{\mathbb{C}}
\newcommand{\Q}{\mathcal{Q}}
\newcommand{\ind}{\mathbf{1}}
\newcommand{\F}{\mathcal{F}}
\newcommand{\FF}{\mathbb{F}}
\newcommand{\G}{\mathcal{G}}
\newcommand{\GG}{\mathbb{G}}
\newcommand{\tildeS}{\widetilde{S}}
\newcommand{\lsi}{\left[\negthinspace\left[}
\newcommand{\rsi}{\right]\negthinspace\right]}

\def\keywordname{{\bf Key words:}}
\newcommand{\keywords}[1]{\par\addvspace\baselineskip\noindent\keywordname\enspace\ignorespaces#1}

\title{A unified approach to pricing and risk management of equity and credit risk}

\makeatletter
\renewcommand*\@fnsymbol[1]{\the#1}
\makeatother

\author{Claudio Fontana\thanks{INRIA Paris - Rocquencourt, Domaine de Voluceau, Rocquencourt, BP 105, 78153, Le Chesnay Cedex (France) and University of Evry Val d'Essonne, Laboratoire Analyse et Probabilités, 23 bd de France, 91037, Evry (France). E-mail: \texttt{claudio.fontana@inria.fr} (corresponding author).} \and Juan Miguel A. Montes\thanks{LMU M\"unchen, Department of Mathematics, Theresienstrasse 39, 80333, M\"unchen (Germany). E-mail: \texttt{juan-miguel.montes@math.lmu.de}.}}

\date{This version: May 19, 2012}

\begin{document}

\maketitle

\begin{abstract}
We propose a unified framework for equity and credit risk modeling, where the default time is a doubly stochastic random time with intensity driven by an underlying affine factor process. This approach allows for flexible interactions between the defaultable stock price, its stochastic volatility and the default intensity, while maintaining full analytical tractability. We characterise all risk-neutral measures which preserve the affine structure of the model and show that risk management as well as pricing problems can be dealt with efficiently by shifting to suitable survival measures. As an example, we consider a jump-to-default extension of the Heston stochastic volatility model.
\end{abstract}

\keywords{default risk, affine processes, stochastic volatility, market price of risk, change of measure, jump-to-default}\vspace{0.15cm}

\noindent \textbf{JEL classification codes:} C02, C60, G12, G33.

\section{Introduction}	\label{S1}

The last few years have witnessed an increasing popularity of hybrid equity/credit risk models, as documented by the recent papers \cite{B,CPS,CL,CM2,CS,CWu,CW,CK}. One of the most appealing features of such models is their capability to link the stochastic behavior of the stock price (and of its volatility) with the randomness of the default event and, hence, with the level of credit spreads. The relation between equity and credit risk is supported by strong empirical evidence (we refer the reader to \cite{CL,CK} for good overviews of the related literature) and several studies document significant relationships between stock volatility and credit spreads of corporate bonds and Credit Default Swaps (\cite{CT,CDMW}).

In this paper, we propose a general framework for the joint modeling of equity and credit risk which allows for a flexible dependence between stock price, stochastic volatility, default intensity and interest rate. The proposed framework is fully analytically tractable, since it relies on the powerful technology of affine processes (see e.g. \cite{DFS,KR} for financial applications of affine processes), and nests several stochastic volatility models proposed in the literature, thereby extending their scope to a defaultable setting. Affine models have been successfully employed in credit risk models, as documented by the papers \cite{CW,Du,GS}. A distinguishing feature of our approach is that, unlike the models proposed in \cite{B,CL,CM2,CS,CW,CK}, we jointly consider both physical and risk-neutral probability measures, ensuring that the analytical tractability is preserved under a change of measure, while at the same time avoiding unnecessarily restrictive specifications of the risk premia. This aspect is of particular importance in credit risk modeling, where one is typically faced with the two problems of computing survival probabilities or related risk measures and of computing arbitrage-free prices of credit derivatives. In this paper, we provide a complete characterisation of the set of risk-neutral measures which preserve the affine structure of the model, thus enabling us to efficiently compute several quantities which are of interest in view of both risk management and pricing applications.

The paper is structured as follows. Section \ref{S2} introduces the modeling framework, while Section \ref{S3} gives a characterisation of the family of risk-neutral measures which preserve the affine structure of the model. In Sections \ref{S4}-\ref{S5}, we show how most quantities of interest for risk management and pricing applications, respectively, can be efficiently computed under suitable (risk-neutral) survival measures (we refer the reader to Sect. 2.5 of \cite{Fo2} for more detailed proofs of the results of Sections \ref{S4}-\ref{S5}). Section \ref{S6} illustrates the main features of the proposed approach within a simple example, which corresponds to a defaultable extension of the Heston \cite{He} model. Finally, Section \ref{S7} concludes.

\section{The modeling framework}	\label{S2}

This section presents the mathematical structure of the modeling framework. Let $(\Omega,\G,P)$ be a reference probability space, with $P$ denoting the physical/statistical probability measure (we want to emphasise that our framework will be entirely formulated with respect to the physical measure $P$). Let $T\in(0,\infty)$ be a fixed time horizon and $W=(W_t)_{0\leq t \leq T}$ an $\R^d$-valued Brownian motion on $(\Omega,\G,P)$, with $d\geq 2$, and denote by $\FF=(\F_t\,)_{0\leq t\leq T}$ its $P$-augmented natural filtration.

We focus our attention on a single defaultable firm, whose \emph{default time} $\tau:\Omega\rightarrow[0,T]\cup\{+\infty\}$ is supposed to be a $(P,\FF)$-\emph{doubly stochastic} random time, in the sense of Def. 9.11 of \cite{McNFE}. This means that there exists a strictly positive $\FF$-adapted process $\lambda^P=(\lambda^P_t)_{0\leq t \leq T}$ such that
$$
P(\tau>t\,|\,\F_T) = P(\tau>t\,|\,\F_t\,) = \exp\,\biggl(-\int_0^t\!\!\lambda^P_udu\biggr)\,,
\qquad \text{ for all }t\in[0,T]\,.
$$
In order to emphasize the role of the reference measure $P$, we call the process $\lambda^P$ the \emph{$P$-intensity} of $\tau$.
Let the filtration $\GG=(\G_t)_{0\leq t \leq T}$ be the \emph{progressive enlargement}\footnote{Due to Lemma 6.1.1 and Lemma 6.1.2 of \cite{BR}, the fact that $P(\tau>t\,|\,\F_T) = P(\tau>t\,|\,\F_t\,)$, for all $t\in[0,T]$, implies that all $(P,\FF)$-martingales are also $(P,\GG)$-martingales. In particular, $W=(W_t)_{0\leq t \leq T}$ is a Brownian motion with respect to both $\FF$ and $\GG$. This important fact will be used in the following without further mention.} of $\FF$ with respect to $\tau$, i.e., $\G_t:=\bigcap_{s>t}\bigl\{\F_s\vee\sigma(\tau\wedge s)\bigr\}$, for all $t\in[0,T]$, and let $\G=\G_T$. It is well-known that $\GG$ is the smallest filtration (satisfying the usual conditions) which makes $\tau$ a $\GG$-stopping time and contains $\FF$, in the sense that $\F_t\subset\G_t$ for all $t\in[0,T]$. 

The price at time $t\in[0,T]$ of one share issued by the defaultable firm is denoted by $S_t\,$. We assume that the $\GG$-adapted process $S=(S_t)_{0\leq t \leq T}$ is continuous and strictly positive on the stochastic interval $\lsi0,\tau\right.\right.\!\!\lsi\right.\right.\!\!$ and satisfies $S\ind_{\lsi\tau,T\rsi}=0$. This means that $S$ drops to zero as soon as the default event occurs and remains thereafter frozen at that level. By relying on Sect. 5.1 of \cite{BR} together with the fact that all $\FF$-martingales are continuous, it can be proved that there exists a continuous strictly positive $\FF$-adapted process $\tildeS=(\tildeS_t)_{0\leq t\leq T}$ such that $S_t=\ind_{\{\tau>t\}}\,\tildeS_t$ holds for all $t\in[0,T]$. We shall refer to the process $\tildeS$ as the \emph{pre-default} value of $S$.

The pre-default value $\tildeS$ is assumed to be influenced by the $\FF$-adapted \emph{stochastic volatility} process $v=(v_t)_{0\leq t \leq T}$ and by an $\R^{d-2}$-valued $\FF$-adapted \emph{factor process} $Y=(Y_t)_{0\leq t \leq T}$. The process $Y$ can include macro-economic covariates describing the state of the economy as well as firm-specific and latent variables, as considered e.g. in \cite{Fo1,FR}. Let us define the process $L=(L_t)_{0\leq t \leq T}$ by $L_t:=\log\tildeS_t$ and the $\R^d$-valued $\FF$-adapted process $X=(X_t)_{0\leq t \leq T}$ by $X_t:=(v_t,Y_t^{\top},L_t)^{\top}$, with $^{\top}$ denoting transposition.

The processes $v$, $Y$ and $L$ are jointly specified through the following square-root-type SDE for the process $X$ on the state space $\R^m_{++}\!\times\R^{d-m}$, where we let $\R^m_{++}:=\{x\in\R^m:x_i>0,\forall i=1,\ldots,m\}$, for some fixed $m\in\{1,\ldots,d-1\}$:
\be	\label{aff-P}
\hspace{2cm}
dX_t = \left(AX_t+b\right)dt+\Sigma\sqrt{R_t}\,dW_t\,,
\qquad\qquad
X_0=\left(v_0,Y_0^{\top},\log S_0\right)^{\top}=\bar{x}\in\R_{++}^m\!\times\R^{d-m}\,,
\ee
where $(A,b,\Sigma)\in\R^{d\times d}\times\R^d\times\R^{d\times d}$ and 
$R_t$ is a diagonal $(d\times d)$-matrix with elements given by $R^{\,i,i}_t=\alpha_i+\beta_i^{\top}X_t$, for all $t\in[0,T]$, with $\alpha:=(\alpha_1,\ldots,\alpha_d)^{\top}\in\R^d_+$ and $\beta:=(\beta_1,\ldots,\beta_d)\in\R^{d\times d}_+$. 

Following the notation adopted in Chapt. 10 of \cite{Fi}, for a given $m\in\{1,\ldots,d-1\}$, we define the sets $I:=\{1,\ldots,m\}$, $J:=\{m+1,\ldots,d\}$ and $D:=I\cup J=\{1,\ldots,d\}$. Intuitively, the set $I$ collects the indices of the first $m$ elements of the $\R^d$-valued process $X$, while the set $J$ collects the remaining ones.
In order to guarantee the existence of a strong solution to the SDE \eqref{aff-P}, we introduce the following assumption.

\begin{Ass}	\label{admissibility}
The parameters $A,b,\Sigma,\alpha,\beta$ satisfy the following conditions:
\begin{enumerate}
\item[(i)]
$b_i\geq (\Sigma_{i,i})^2\beta_{i,i}/2$ for all $i\in I$;
\item[(ii)]
$A_{i,j}=0$ for all $i\in I$ and $j\in J$ and $A_{i,j}\geq 0$ for all $i,j\in I$ with $i\neq j$;
\item[(iii)]
$\Sigma_{i,j}=0$ for all $i\in I$ and $j\in D$ with $j\neq i$;
\item[(iv)]
$\beta_{j,i}=0$ for all $i\in D$ and $j\in J$, $\beta_{i,i}>0$ for all $i\in I$ and $\beta_{i,j}=0$ for all $i,j\in I$ with $i\neq j$;
\item[(v)]
$\alpha_i=0$ for all $i\in I$ and $\alpha_{\!j}>-\sum_{i\in I}\beta_{i,j}$ for all $j\in J$.
\end{enumerate}
\end{Ass}

For any $\bar{x}\in\R^m_{++}\!\times\R^{d-m}$, Assumption \ref{admissibility} ensures the existence of a unique strong solution $X=(X_t)_{0\leq t \leq T}$ to the SDE \eqref{aff-P} on the filtered probability space $(\Omega,\G,\FF,P)$ such that $X_0=\bar{x}$ and $X_t\in\R^m_{++}\!\times\R^{d-m}$ $P$-a.s. for all $t\in[0,T]$. Indeed, the same arguments used in the proof of Lemma 10.6 of \cite{Fi} give the existence of a unique strong solution $X=(X_t)_{0\leq t \leq T}$ on $\R^m_+\times\R^{d-m}$, while Lemma A.3 of \cite{DK} together with Ex. 10.12 of \cite{Fi} implies that $X$ actually takes values in $\R^m_{++}\!\times\R^{d-m}$.
Due to conditions (iv)-(v) of Assumption \ref{admissibility}, this also implies that the matrix $R_t$ is positive definite for all $t\in[0,T]$. In the remaining part of the paper, we shall always assume that Assumption \ref{admissibility} is satisfied without further mention.

\begin{Rem}
The parameter restrictions imposed by Assumption \ref{admissibility} bear resemblance to the canonical representation of \cite{DSi}. However, we do not require the matrix $\Sigma$ to be diagonal, since this may lead to unnecessary restrictions on the model if $2\leq m\leq d-2$, as pointed out in \cite{CFK2}.
\end{Rem}

The following proposition describes the dynamics of the defaultable stock price process $S$.

\begin{Prop}	\label{SDE-S}
The process $S=(S_t)_{0\leq t \leq T}$ satisfies the following SDE on $(\Omega,\G,\GG,P)$:
\be	\label{dynamics-S}	\ba
dS_t &= S_{t-}\left(\bar{s}+\mu_1\!\log S_{t-}+\mu_2v_t
+\sum_{i=1}^{d-2}\eta_iY_t^i+\sum_{i=1}^{m-1}\bar{\eta}_iY_t^i\right)dt	\\
&\quad +S_{t-}\,\sigma\sqrt{v_t}\,dW^1_t
+S_{t-}\sum_{i=2}^d\Sigma_{d,i}\sqrt{R^{\,i,i}_t}dW^i_t
-S_{t-}\,d\mathbf{1}_{\{\tau\leq t\}}
\ea	\ee
with the convention $S_{t-}\log S_{t-}=0$ on $\{\tau\leq t\}$ and where
$$ \begin{aligned}
\bar{s} &:= b_d+\frac{1}{2}\sum_{k=m+1}^d\left(\Sigma_{d,k}\right)^2\alpha_k\,,
&\,
\mu_1 := A_{d,d}\,,
\quad 
\mu_2 &:= A_{d,1}+\frac{1}{2}\left(\Sigma_{d,1}\right)^2\beta_{1,1}+\frac{1}{2}\sum_{k=m+1}^d\left(\Sigma_{d,k}\right)^2\beta_{1,k}\,,	\\
\eta_i &:= A_{d,i+1},
\,
&\sigma := \Sigma_{d,1}\sqrt{\beta_{1,1}}\,,
\quad
\bar{\eta}_i &:= \frac{1}{2}\left(\Sigma_{d,i+1}\right)^2\beta_{i+1,i+1}+\frac{1}{2}\sum_{k=m+1}^{d}\left(\Sigma_{d,k}\right)^2\beta_{i+1,k}\,.
\end{aligned} $$
\end{Prop}
\begin{proof}
Observe first that $dS_t=\ind_{\{\tau>t-\}}\tildeS_{t-}\,\bigl(dL_t+d\langle L\rangle_t/2\bigr)-\tildeS_{t-}\,d\ind_{\{\tau\leq t\}}$, due to It\^o's formula and integration by parts. Equation \eqref{dynamics-S} then follows from \eqref{aff-P} together with Assumption \ref{admissibility} by means of simple computations.
\end{proof}

\begin{Rem}
As can be seen from Proposition \ref{SDE-S}, the defaultable price process $S$ has a rich structure, influenced by the factor process $Y$ in both the drift and diffusion terms. Furthermore, there are three levels of dependence between $S$ and the stochastic volatility $v$: \emph{(1)} a direct interaction, since $v$ explicitly appears in the dynamics of $S$; \emph{(2)} a semi-direct interaction, since the Brownian motion $W^1$ driving the process $v$ is also one of the drivers of $S$; \emph{(3)} an indirect interaction, since $S$ and $v$ both depend on the factor process $Y$.
\end{Rem}

To complete the description of the modeling framework, we specify as follows the $P$-intensity process $\lambda^P=(\lambda^P_t)_{0\leq t \leq T}$ and the risk-free interest rate process $r=(r_t)_{0\leq t \leq T}$:
\be	\label{lambda-r}
\lambda^P_t := \bar{\lambda}^P+(\Lambda^P)^{\top}X_t\,,
\qquad\quad
r_t := \bar{r}+\Upsilon^{\top}X_t\,,
\qquad\quad
\text{ for all }t\in[0,T]\,,
\ee
where the parameters $\bar{\lambda}^P,\bar{r}\in\R_+$ and $\Lambda^P,\Upsilon\in\R^m_+\times\{0\}^{d-m}$ satisfy $\bar{\lambda}^P+\sum_{i=1}^m\Lambda^P_i>0$ and $\bar{r}+\sum_{i=1}^m\Upsilon_i>0$. This ensures that the $P$-intensity and the risk-free rate are correlated and strictly positive, since $0$ is an unattainable boundary for $X^i$, $\forall i\in I$. Furthermore, the linear structure \eqref{lambda-r} permits to obtain analytically tractable formulae for several quantities of interest, as shown in Sections \ref{S4}-\ref{S5}. The specification \eqref{lambda-r} allows for a direct dependence of $\lambda^P$ on the stochastic volatility $v$, this feature being consistent with several empirical observations (see e.g. \cite{CT,CDMW}). Furthermore, the defaultable price process $S$ and the $P$-intensity $\lambda^P$ are linked through the common factor process $Y$. Finally, we want to remark that the proposed modeling framework generalises to a defaultable setting several stochastic volatility models considered in the literature. For instance, defaultable versions of the models considered in \cite{AR,CS} and Sect. 4.3 of \cite{DPS} can be easily recovered within our general setting. As an example, in Section \ref{S6} we shall study in detail an extended defaultable version of the Heston \cite{He} stochastic volatility model.

\begin{Rem}
We want to point out that multifactor stochastic volatility models are naturally embedded within our modeling framework. Indeed, the first $m-1$ components of the factor process $Y$ are strictly positive processes and can be interpreted as additional stochastic volatility factors, as can also be seen from equation \eqref{dynamics-S}. For instance, in the case $d=3$ and $m=2$, we can easily obtain (a defaultable version of) the two-factor stochastic volatility model proposed by \cite{CHJ}.
\end{Rem}

\begin{Rem}
The modeling framework described in this section can be easily extended to the case of $M>1$ defaultable firms if we suppose that their random default times $\{\tau_1,\ldots,\tau_M\}$ are $\FF$-conditionally independent (see \cite{McNFE}, Sect. 9.6).
In that case, the process $L$ is an $\R^M$-valued process representing the logarithm of the pre-default values of the $M$ stock prices (and, similarly, the process $v$ representing the stochastic volatilities of the $M$ stocks is also $\R^M$-valued). If the processes $L$, $v$ and the factor process $Y$ are jointly modeled as an affine diffusion of the type \eqref{aff-P} and if the $P$-intensity processes $\lambda^{P,\ell}=(\lambda^{P,\ell}_t)_{0\leq t \leq T}$, for $\ell=1,\ldots,M$, are of the form \eqref{lambda-r}, then the multi-firm extension of the model is still fully analytically tractable. This generalization can be of particular interest in view of portfolio credit risk modeling.
\end{Rem}

\section{Equivalent changes of measure which preserve the affine structure}	\label{S3}

The modeling framework introduced in Section \ref{S2} has been formulated entirely with respect to the physical probability measure $P$. However, since we aim at dealing with pricing as well as risk management applications, we need to study the structure of the model under a suitable \emph{risk-neutral} probability measure, formally defined as a probability measure $Q\sim P$ on $(\Omega,\G)$ such that the discounted defaultable price process $\exp\,\bigl(-\int_0^{\cdot}\!r_udu\bigr)\,S$ is a $(Q,\GG)$-local martingale\footnote{Due to the fundamental result of \cite{DSc}, this is equivalent to the validity of \emph{No Free Lunch with Vanishing Risk} (NFLVR) condition for the financial market $(S,\GG)$, being the process $\exp\bigl(-\int_0^{\cdot}\!r_udu\bigr)\,S$ locally bounded. In particular, this excludes the existence of arbitrage opportunities.}.

It is important to be aware of the fact that most of the appealing features of the framework described in Section \ref{S2} may be lost after a change of measure. Aiming at a model which is analytically tractable under both the physical and a risk-neutral measure, we shall consider the family of risk-neutral measures $Q$ which \emph{preserve the affine structure of $(X,\tau)$}, in the sense of the following definition.

\begin{Def}	\label{aff-pres}
Let $Q$ be a probability measure on $(\Omega,\G)$ with $Q\sim P$. We say that $Q$ \emph{preserves the affine structure of $(X,\tau)$} if the following hold:
\vspace{-0.08cm}
\begin{enumerate}
\item[(i)]
the process $X=(X_t)_{0\leq t \leq T}$ satisfies an SDE of the type \eqref{aff-P} on $(\Omega,\G,\FF,Q)$ with respect to an $\R^d$-valued $(Q,\FF)$-Brownian motion $W^Q=(W^Q_t)_{0\leq t \leq T}$ and for some parameters $A^Q,b^Q,\Sigma,\alpha$, $\beta$ satisfying Assumption \ref{admissibility};
\item[(ii)]
the default time $\tau$ is a $(Q,\FF)$-doubly stochastic random time with $Q$-intensity $\lambda^Q=(\lambda^Q_t)_{0\leq t \leq T}$ of the form $\lambda^Q_t=\bar{\lambda}^Q+(\Lambda^Q)^{\top}X_t$, for $\bar{\lambda}^Q\in\R_+$ and $\Lambda^Q\in\R^m_+\times\{0\}^{d-m}$ with $\bar{\lambda}^Q+\sum_{i=1}^m\Lambda^Q_i>0$.
\end{enumerate}
\end{Def}

We denote by $\Q$ the family of all risk-neutral measures which preserve the affine structure of $(X,\tau)$, in the sense of Definition \ref{aff-pres}. The next theorem gives a complete characterisation of the family $\Q$. This result follows from a more general one in Chapt. 2 of \cite{Fo2}, but we outline a self-contained proof for the convenience of the reader. We denote by $\mathcal{E}$ the stochastic exponential and by $M\!=\!(M_t)_{0\leq t \leq T}$ the $(P,\GG)$-martingale defined by $M_t:=\ind_{\{\tau\leq t\}}\!-\!\int_0^{t\wedge\tau}\!\!\lambda^P_udu$ (see \cite{BR}, Prop. 5.1.3).

\begin{Thm}	\label{measure-change}
Let $Q$ be a probability measure on $(\Omega,\G)$. Then we have $Q\in\Q$ if and only if
\be	\label{RN}	\ba
\frac{dQ}{dP} 
&= \mathcal{E}\left(\int\!\theta\,dW+\int\!\gamma\,dM\right)_{\!T}	\\
&= \exp\left(\sum_{i=1}^d\int_0^T\!\!\theta_t^i\,dW^i_t
-\frac{1}{2}\sum_{i=1}^d\int_0^T\!(\theta^i_t)^2dt
-\int_0^{\tau\wedge T}\!\!\!\gamma_t\,\lambda^P_tdt\right)
\Bigl(1+\mathbf{1}_{\left\{\tau\leq T\right\}}\gamma_{\tau}\Bigr)
\ea	\ee
where $\theta=(\theta_t)_{0\leq t \leq T}$ and $\gamma=(\gamma_t)_{0\leq t \leq T}$ are $\FF$-adapted processes of the following form:
\be	\label{MPR-1}
\theta_t = \theta(X_t) := R_t^{-1/2}\,\bigl(\hat{\theta}+\Theta X_t\bigr)\,,
\qquad\qquad
\gamma_t = \gamma(X_t)
:= \frac{\bigl(\bar{\lambda}^Q-\bar{\lambda}^P\bigr)+\bigl(\Lambda^Q-\Lambda^P\bigr)^{\top}X_t}
{\bar{\lambda}^P+\bigl(\Lambda^P\bigr)^{\top}X_t}\,,
\ee
for some $\hat{\theta}\in\R^d$ and $\Theta\in\R^{d\times d}$ such that:
\begin{enumerate}
\item[(i)]
$\sum_{k=1}^d\Sigma_{i,k}\hat{\theta}_k\geq (\Sigma_{\,i,i})^2\beta_{i,i}/2-b_i$ for all $i\in I$;
\item[(ii)]
$\sum_{k=1}^d\Sigma_{i,k}\Theta_{k,j}=0$, for all $i\in I$ and $j\in J$, and $\sum_{k=1}^d\Sigma_{i,k}\Theta_{k,j}\geq-A_{i,j}$, for all $i,j\in I$ with $i\neq j$;
\end{enumerate}
for some $\bar{\lambda}^Q\in\R_+$ and $\Lambda^Q\in\R^m_+\times\{0\}^{d-m}$ with $\bar{\lambda}^Q+\sum_{i=1}^m\Lambda^Q_i>0$
and if the following equality holds $P$-a.s. on $\{\tau>t\}$, using the notation introduced in Proposition \ref{SDE-S}:
\be	\label{MPR-3}
\bar{s}+\mu_1\log S_{t-}+\biggl(\mu_2+\sigma\frac{\theta^1_t}{\sqrt{v_t}}\biggr)\,v_t+\sum_{i=1}^{d-2}\eta_iY_t^i+\sum_{i=1}^{m-1}\bar{\eta}_iY_t^i+\sum_{i=2}^d\Sigma_{d,i}\sqrt{R_t^{\,i,i}}\,\theta_t^i 
= r_t + \lambda^P_t(1+\gamma_t)\,.
\ee
\end{Thm}
\begin{proof}
Let $\theta=(\theta_t)_{0\leq t \leq T}$ and $\gamma=(\gamma_t)_{0\leq t \leq T}$ be two $\FF$-adapted processes satisfying \eqref{MPR-1}. Since $\theta$ and $\gamma$ are continuous functions of $X$ and the process $X$ is continuous, hence locally bounded, the process $Z:=\mathcal{E}\bigl(\int\!\theta\,dW+\int\!\gamma\,dM\bigr)$ is well-defined as a strictly positive $(P,\GG)$-local martingale and, as a consequence of Fatou's lemma, it is also a $(P,\GG)$-supermartingale. Moreover, Thm. 2.4 and Remark 2.5 of \cite{CFY} allow to conclude that $E[Z_T]=1$, thus implying that the process $Z$ is a uniformly integrable $(P,\GG)$-martingale. So, we can define a probability measure $Q$ on $(\Omega,\G)$ via \eqref{RN}. Part (i) of Definition \ref{aff-pres} then follows from Girsanov's theorem together with \eqref{MPR-1}, while part (ii) follows from Thm. 6.3 of \cite{CJN}, Girsanov's theorem together with \eqref{MPR-1} and Prop. 6.2.2 of \cite{BR}. Finally, the $(Q,\GG)$-local martingale property of $\exp\bigl(-\int_0^{\cdot}\!r_udu\bigr)\,S$ easily follows from Girsanov's theorem together with Proposition \ref{SDE-S} and equation \eqref{MPR-3}.
Conversely, suppose that $Q\in\Q$. The existence of a representation of the form \eqref{RN} follows from Cor. 5.2.4 of \cite{BR}, while \eqref{MPR-1} and \eqref{MPR-3} follow from Girsanov's theorem together with Definition \ref{aff-pres} and Proposition \ref{SDE-S}, respectively.
\end{proof}

Note that the process $\gamma=(\gamma_t)_{0\leq t \leq T}$ introduced in \eqref{MPR-1} satisfies $\gamma_t>-1$ $P$-a.s. for all $t\in[0,T]$, due to the restrictions imposed on the parameters $\bar{\lambda}^P$, $\bar{\lambda}^Q$, $\Lambda^P$ and $\Lambda^Q$. In particular, this ensures that, for every probability measure $Q\in\Q$, both the $P$-intensity process $\lambda^P=(\lambda^P_t)_{0\leq t \leq T}$ and the $Q$-intensity process $\lambda^Q=(\lambda^Q_t)_{0\leq t \leq T}$ are $P$-a.s. strictly positive.

Due to Theorem \ref{measure-change}, the preservation of the affine structure of $(X,\tau)$ does not prevent the default intensity to change significantly from the physical to a risk-neutral probability measure $Q\in\Q$, due to the presence of the risk premium $\gamma$ (see also the comments below). From the practical perspective, this is an important aspect of our modeling approach, especially in view of the possibility of valuing credit/equity financial derivatives whose payoff also depends on the $P$-intensity of default through, for instance, the rating score attached to a defaultable firm or the corresponding statistical survival/default probability.

\begin{Rem}	\label{risk-premia}
The processes $\theta=(\theta_t)_{0\leq t \leq T}$ and $\gamma=(\gamma_t)_{0\leq t \leq T}$ admit the financial interpretation of \emph{risk premia} (or \emph{market prices of risk}) associated to the randomness generated by the Brownian motion $W$ and by the random default time $\tau$, respectively. More specifically:
\begin{enumerate} 
\item[(a)]
The process $\theta=(\theta_t)_{0\leq t \leq T}$ represents the risk premium associated to the diffusive risk generated by the Brownian motion $W$. Since the stock price, its stochastic volatility, the default intensity and the interest rate all depend on $W$ through the process $X$, the risk premium $\theta$ can be considered as a market-wide non-diversifiable risk premium\footnote{In the context of default-free term structure modeling, in \cite{CFK1} the authors demonstrate that the specification \eqref{MPR-1} has a considerably better fit to market data than the simpler market price of risk specifications traditionally considered in the literature (see e.g. \cite{CWu,DSi,Du1,Du2,He}).}.
\item[(b)]
The process $\gamma=(\gamma_t)_{0\leq t \leq T}$ represents the risk premium associated to the default event or, more precisely, the risk premium associated to the idiosyncratic component of the risk generated by the occurrence of the default event (to this effect, see also \cite{CPS,EKM} and Sect. 9.3 of \cite{McNFE}). 
\end{enumerate}
The importance of explicitly distinguishing between $\theta$ and $\gamma$ has been demonstrated in \cite{Dr}. Assuming $\gamma\equiv 0$ means that the idiosyncratic component of default risk can be diversified away in the market, as explained in \cite{JLY}, and, therefore, market participants do not require a compensation for it. However, the jump-type risk premium can be significant when it is difficult to hedge the risk associated with the timing of the default event of a given firm. Note that, as can be seen from \eqref{MPR-1}, the risk premia $\theta$ and $\gamma$ both depend on the common driving process $X$.
\end{Rem}

Due to Theorem \ref{measure-change}, our modeling framework enjoys full analytical tractability under both the physical measure $P$ and any risk-neutral measure $Q\in\Q$, thus enabling us to efficiently solve risk management as well as a pricing problems, as we are going to show in Sections \ref{S4}-\ref{S5}. We close this section with the following fundamental result, which follows from Thm. 10.4 of \cite{Fi} together with part (i) of Definition \ref{aff-pres}, equation \eqref{aff-P} and Assumption \ref{admissibility}. For $z\in\C^d$ we denote by $\Re(z)$ and $\Im(z)$ the real and imaginary parts of $z$, respectively, and $\C^m_-:=\bigl\{z\in\C^m:\Re(z)\in\R^m_-\bigr\}$. For $Q\in\Q\cup\{P\}$, we denote by $E^Q$ the (conditional) expectation operator under the measure $Q$, with $E:=E^P$.

\begin{Prop}	\label{CF}
For every $Q\in\Q\cup\{P\}$ and for all $z\in\C^m_-\times i\R^{d-m}$, there exists a unique solution $\bigl(\Phi^Q(\cdot,z),\Psi^Q(\cdot,z)\bigr):\left[0,T\right]\rightarrow\C\times\C^d$ to the following system of Riccati ODEs:
\be	\label{Ric}	\ba
\partial_t\Phi^Q(t,z)
&= (b^Q)^{\!\top}\,\Psi^Q(t,z)+\frac{1}{2}\sum_{k=m+1}^d[\Sigma^{\top}\Psi^Q(t,z)]_k^2\,\alpha_k
-\bar{\lambda}^Q - \bar{r}\,\ind_{\!Q\neq P}\,, \\
\Phi^Q(0,z) &= 0\,, \\
\partial_t\Psi^Q_i(t,z)
&= \sum_{k=1}^dA^Q_{k,i}\,\Psi^Q_k(t,z)
+\frac{1}{2}[\Sigma^{\top}\Psi^Q(t,z)]_i^2\beta_{i,i}
+\frac{1}{2}\sum_{k=m+1}^d[\Sigma^{\top}\Psi^Q(t,z)]_k^2\,\beta_{i,k} - \Lambda^Q_i	
- \Upsilon_i\ind_{\!Q\neq P}\,,\\[-0.25cm]
&\hspace{13.85cm} \forall i\in I\,,\\
\partial_t\Psi^Q_j(t,z) 
&= \sum_{k=m+1}^dA^Q_{k,j}\,\Psi^Q_k(t,z)\,,	
\hspace{9.7cm} \forall j\in J\,,\\
\Psi^Q(0,z) &= z\,.
\ea	\ee
Furthermore, for every $Q\in\Q\cup\{P\}$, the following holds for all $0\leq t\leq u\leq T$ and for all $z\in\C^m_-\times i\R^{d-m}$:
\be	\label{CF-1}
E^Q\left[\exp\biggl(-\!\int_t^u\!\!(\lambda^Q_s+r_s\ind_{Q\neq P})\,ds+z^{\top}X_u\biggr)\biggr|\,\F_t\right]
= \exp\left(\Phi^Q\left(u-t,z\right)+\Psi^Q\left(u-t,z\right)^{\top}X_t\right).
\ee
\end{Prop}

\section{Risk management applications}	\label{S4}

Many quantities of interest in view of risk management applications can be computed as conditional expectations under the physical measure $P$. As a first and basic application, let us compute the $\G_t$-conditional survival probability of the defaultable firm up to the final horizon $T$. We denote by $\Phi^P(\cdot,\cdot)$ and $\Psi^P(\cdot,\cdot)$ the solutions to the Riccati ODEs \eqref{Ric} with $Q=P$.

\begin{Prop}	\label{surv-P}
For any $t\in[0,T]$ the following holds:
\be	\label{surv-prob}
P\left(\tau>T\,|\,\G_t\right) = \ind_{\{\tau>t\}}\exp\left(\Phi^P(T-t,0)+\Psi^P(T-t,0)^{\!\top}X_t\right).
\ee
\end{Prop}
\begin{proof}
Cor. 5.1.1 of \cite{BR} implies that $P(\tau>T\,|\,\G_t)=\ind_{\{\tau>t\}}E[\exp(-\int_t^T\!\!\lambda^P_sds)\,|\,\F_t\,]$. The result then follows by applying formula \eqref{CF-1} with $Q=P$, $z=0$ and $u=T$.
\end{proof}

As can be easily checked from \eqref{Ric}, the right-hand side of \eqref{surv-prob} only depends on $\{X^i:i\in I\}$, i.e., on the components of the process $X$ on which the $P$-intensity $\lambda^P$ depends. 
For computing conditional expectations (under the measure $P$) of more general quantities needed for risk management purposes, it turns out to be convenient to introduce the \emph{$T$-survival measure} $P^T\sim P$ on $(\Omega,\G)$ defined by $dP^T/dP:=\exp\bigl(-\int_0^T\!\!\lambda^P_tdt\bigr)/E\bigl[\exp\bigl(-\int_0^T\!\!\lambda^P_tdt\bigr)\bigr]$.

\begin{Lem}	\label{PT}
For any random variable $F\in L^1(P,\F_T)$ and for any $t\in[0,T]$ the following holds:
\be	\label{PT-1}
E\left[F\,\ind_{\{\tau>T\}}\,|\,\G_t\right] 
= P(\tau>T\,|\,\G_t)\,E^{P^T}\!\left[F\,|\,\F_t\,\right].
\ee
\end{Lem}
\begin{proof}
Cor. 5.1.1 of \cite{BR} implies that $E\left[F\,\ind_{\{\tau>T\}}\,|\,\G_t\right]=\ind_{\{\tau>t\}}E\bigl[F\exp\bigl(-\int_t^T\!\!\lambda^P_sds\bigr)\,|\,\F_t\,\bigr]$. Equation \eqref{PT-1} then follows by using the definition of the measure $P^T$ together with the conditional version of Bayes' formula (see e.g. \cite{Fi}, Ex. 4.9).
\end{proof}

Lemma \ref{PT} shows that the computation of the $\G_t$-conditional expectation of an $\F_T$-measurable random variable $F$ in the case of survival up to time $T$ reduces to the computation of the $\F_t$-conditional expectation of $F$ under the $T$-survival measure $P^T$, the term $P(\tau>T\,|\,\G_t)$ being given as in \eqref{surv-prob}. As can be seen from equation \eqref{PT-1}, the $T$-survival measure $P^T$ allows to decompose the conditional expectation of the product $F\,\ind_{\{\tau>T\}}$ into the product of two conditional expectations. Note also that, from the point of view of practical applications, the term $P(\tau>T\,|\,\G_t)$ does not necessarily have to be computed, since it can often be deduced from publicly available data, notably from rating transition matrices published by rating agencies. Furthermore, as shown in the next lemma, the $\F_t$-conditional characteristic function of the vector $X_T$ under the $T$-survival measure $P^T$ can be computed in closed form.

\begin{Lem}	\label{CF-P}
For any $z\in i\R^d$ and for any $t\in[0,T]$ the following holds:
\be	\label{CF-P1}
\varphi_t^{P^T}\!(z) :=
E^{P^T}\bigl[e^{\,z^{\top}X_T}|\,\F_t\,\bigr]
= \exp\,\Bigl(\Phi^P(T-t,z)-\Phi^P(T-t,0)+\bigl(\Psi^P(T-t,z)-\Psi^P(T-t,0)\bigr)^{\top}X_t\Bigr).
\ee
\end{Lem}
\begin{proof}
The definition of the $T$-survival measure $P^T$ together with the conditional version of Bayes' formula gives $E^{P^T}\bigl[e^{\,z^{\top}X_T}|\,\F_t\,\bigr]=E\bigl[\exp\bigl(-\int_t^T\!\!\lambda^P_sds+z^{\top}X_T\bigr)\,|\,\F_t\,\bigr]/E\bigl[\exp\bigl(-\int_t^T\!\!\lambda^P_sds\bigr)\,|\,\F_t\,\bigr]$. By applying \eqref{CF-1} with $Q=P$, $u=T$ and $z\in i\R^d$ ($z=0$, resp.) to the numerator (to the denominator, resp.), we then obtain equation \eqref{CF-P1}.
\end{proof}

Due to Lemma \ref{PT} and Lemma \ref{CF-P}, we can compute the $\G_t$-conditional expectation (under the physical probability measure $P$) of arbitrary functions of the random vector $X_T$ in the case of survival by relying on well-known Fourier inversion techniques. As an example, we can explicitly compute quantiles of the $\mathcal{G}_t$-conditional distribution of the defaultable price $S_T$ in the case of survival. This is crucial for the computation of \emph{Value-at-Risk} and related risk measures.

\begin{Prop}	\label{distr-fct}
For any $x\in(0,\infty)$ and for any $t\in[0,T]$ the following holds:
\be	\label{distr-fct-1}
P(S_T\leq x,\tau>T\,|\,\G_t) 
= P\left(\tau>T\,|\,\G_t\right)
\left(\frac{1}{2}-\frac{1}{\pi}
\int_0^{\infty}\frac{\Im\bigl(e^{-iy\log x}\,\varphi^{P^T}_t\!(0,\ldots,0,iy)\bigr)}{y}\,dy\right)
\ee
where $P(\tau>T\,|\,\G_t)$ and $\varphi_t^{P^T}\!(\cdot)$ are explicitly given in equation \eqref{surv-prob} and Lemma \ref{CF-P}, respectively.
\end{Prop}
\begin{proof}
Note that 
$$
P\left(S_T\leq x,\tau>T|\G_t\right)
= P\left(L_T\leq\log x, \tau>T|\G_t\right)
= P\left(\tau>T|\G_t\right)P^T\!\left(L_T\leq\log x|\F_t\right)
$$
where the second equality follows from Lemma \ref{PT}. Equation \eqref{distr-fct-1} then follows from standard Fourier inversion techniques (see e.g. \cite{Fo2}, Prop. 2.5.12, and \cite{Pa}, Sect. 1.2.6).
\end{proof}

\section{Valuation of default-sensitive payoffs and defaultable options}	\label{S5}

Throughout this section, we fix an element $Q\in\Q$. For the purpose of valuing default-sensitive payoffs, the \emph{$u$-survival risk-neutral measure} $Q^u$, for $u\in[0,T]$, turns out to be quite useful. The measure $Q^u$ is defined by $dQ^u/dQ=\exp\bigl(-\int_0^u(r_s+\lambda^Q_s)\,ds\bigr)/E^Q\bigl[\exp\bigl(-\int_0^u(r_s+\lambda^Q_s)\,ds\bigr)\bigr]$. For $u=T$, the measure $Q^T$ bears resemblance to the \emph{$T$-survival measure} $P^T$ introduced in Section \ref{S4}, except that $Q^T$ is defined with respect to some $Q\in\Q$ and the density $dQ^T/dQ$ also involves the risk-free interest rate besides the $Q$-intensity $\lambda^Q$ (compare also with \cite{BR}, Def. 15.2.2). Following the same logic of Section \ref{S4}, we show that many pricing problems can be simplified by shifting to the measure $Q^u$, for some $u\in[0,T]$. As a preliminary, let us compute the arbitrage-free price $\Pi(t,T)$ of a \emph{zero-coupon defaultable bond}. We denote by $\Phi^Q(\cdot,\cdot)$ and $\Psi^Q(\cdot,\cdot)$ the solutions to the Riccati ODEs \eqref{Ric}. The proof of the following lemma is completely analogous to that of Proposition \ref{surv-P} but we include it for the convenience of the reader.

\begin{Lem}	\label{surv-Q}
For any $t\in[0,T]$ the following holds:
\be	\label{surv-Q1}
\Pi(t,T) = \ind_{\{\tau>t\}}\exp\left(\Phi^Q(T-t,0)+\Psi^Q(T-t,0)^{\!\top}X_t\right).
\ee
\end{Lem}
\begin{proof}
Note first that 
$$
\Pi(t,T)
=E^Q\left[e^{-\int_t^T\!\!r_sds}\,\ind_{\{\tau>T\}}\bigr|\G_t\right]
=\ind_{\{\tau>t\}}E^Q\left[e^{-\int_t^T\!(r_s+\lambda^Q_s)\,ds}\bigr|\F_t\right]
$$
where the second equality follows from Thm. 9.23 of \cite{McNFE}. Equation \eqref{surv-Q1} then follows from Proposition \ref{CF} with $u=T$ and $z=0$.
\end{proof}

Of course, coupon-bearing corporate bonds can be valued as linear combinations of zero-coupon defaultable bonds (see \cite{BR}, Sect. 1.1.5). More generally, most default-sensitive payoffs can be decomposed into linear combinations of \emph{zero-recovery} and \emph{pure recovery} payments, the latter being paid only in the case of default, see e.g. Sect. 9.4 of \cite{McNFE}. The next proposition provides general valuation formulas for zero-recovery and pure recovery payments.

\begin{Prop}	\label{gen}
For any $t\in[0,T]$ and for any measurable function $G:\R^m_{++}\!\times\R^{d-m}\rightarrow\R_+$ the following hold:
\be	\label{gen-1}
E^Q\Bigl[e^{-\int_t^T\!\!r_sds}G(X_T)\ind_{\{\tau>T\}}\bigr|\,\G_t\Bigr]
= \Pi(t,T)\,E^{Q^T}[G(X_T)\,|\,\F_t\,]\,,
\ee
\be	\label{gen-2}
\ind_{\{\tau>t\}}E^Q\Bigl[e^{-\int_t^{\tau}\!\!r_sds}G(X_{\tau})\ind_{\{\tau\leq T\}}\bigr|\,\G_t\Bigr]
= \int_t^T\!\!\Pi(t,u)\,E^{Q^u}[\lambda^Q_u\,G(X_u)\,|\,\F_t\,]\,du\,.
\ee
\end{Prop}
\begin{proof}
Note first that, due to Thm. 9.23 of \cite{McNFE}, we can write:
$$	\begin{aligned}
E^Q\Bigl[e^{-\int_t^T\!\!r_sds}G(X_T)\ind_{\{\tau>T\}}\bigr|\,\G_t\Bigr]
&= \ind_{\{\tau>t\}}E^Q\Bigl[e^{-\int_t^T\!(r_s+\lambda^Q_s)\,ds}G(X_T)\,\bigr|\,\F_t\,\Bigr]\,,	\\
\ind_{\{\tau>t\}}E^Q\Bigl[e^{-\int_t^{\tau}\!\!r_sds}G(X_{\tau})\ind_{\{\tau\leq T\}}\bigr|\,\G_t\Bigr]
&= \ind_{\{\tau>t\}}E^Q\biggl[\int_t^T\!\!e^{-\int_t^u(r_s+\lambda^Q_s)\,ds}\lambda^Q_u\,G(X_u)\,du\,\Bigr|\,\F_t\,\biggr]\,.
\end{aligned} $$
Equations \eqref{gen-1}-\eqref{gen-2} then follow by using the definition of the measure $Q^u$, for $u\in[t,T]$, together with the conditional version of Bayes' formula and also, for \eqref{gen-2}, Tonelli's theorem.
\end{proof}

We want to point out that, in view of practical applications, the quantities $\Pi(t,u)$, for $u\in[t,T]$, appearing in equations \eqref{gen-1}-\eqref{gen-2} do not necessarily have to computed, since they can be directly observed on the corporate bond market. This fact represents one of the main advantages of using survival risk-neutral measures for the valuation of defaultable claims (see also \cite{Sch} for a related discussion and other applications of survival measures to credit risk modeling).

As an application of Lemma \ref{surv-Q} and Proposition \ref{gen}, we compute the fair spread $\pi^{\,\text{CDS}}(t,T)$, at time $t\in[0,T]$, of a \emph{Credit Default Swap (CDS)} which exchanges a fixed stream of payments in arrears equal to $\pi^{\,\text{CDS}}(t,T)$ at the dates $\{t_1,\ldots,t_N\}$, with $t\leq t_1<\ldots<t_N\leq T$, (\emph{premium payment leg}) against the payment at the default time $\tau$ (if the latter happens before the maturity $T$) of a default protection term equal to a fraction $\delta\in(0,1)$ of the unitary nominal value (\emph{default payment leg}), see e.g. Sect. 9.3 of \cite{McNFE}.

\begin{Cor}	\label{CDS}
For any $t\in[0,T]$ and $t_0:=t\leq t_1<\ldots<t_N\leq T$ the following holds on $\{\tau>t\}$:
\be	\label{spread-CDS}
\pi^{\,\text{\emph{CDS}}}(t,T)
= \delta\,\frac{\int_t^T\!\Pi(t,u)\,E^{Q^u}[\lambda^Q_u\,|\,\F_t\,]\,du}
{\sum_{k=1}^N(t_k-t_{k-1})\,\Pi(t,t_k)}\,.
\ee
\end{Cor}
\begin{proof}
Due to Lemma \ref{surv-Q}, the arbitrage-free price of the premium payment leg is given by:
$$
\pi^{\,\text{CDS}}(t,T)\sum_{k=1}^N\bigl(t_k-t_{k-1}\bigr)
E^Q\Bigl[e^{-\int_t^{t_k}\!r_sds}\ind_{\{\tau>t_k\}}\bigr|\,\G_t\Bigr]
= \pi^{\,\text{CDS}}(t,T)\sum_{k=1}^N(t_k-t_{k-1})\,\Pi(t,t_k).
$$
On the other hand, due to equation \eqref{gen-2}, the arbitrage-free price of the default payment leg is equal to:
$$
\ind_{\{\tau>t\}}E^Q\left[e^{-\int_t^{\tau}\!\!r_sds}\delta\,\ind_{\{\tau\leq T\}}\bigr|\,\G_t\right]
= \delta\!\int_t^T\!\!\Pi(t,u)\,E^{Q^u}[\lambda^Q_u\,|\,\F_t\,]\,du.
$$
Equation \eqref{spread-CDS} then follows by recalling that, by definition, the fair spread $\pi^{\,\text{CDS}}(t,T)$ is the premium payment which equates the values of the two legs of the CDS (see \cite{McNFE}, Sect. 9.3).
\end{proof}

For $0\leq t \leq u\leq T$, the next lemma gives the explicit expression of the $\F_t$-conditional characteristic function $\varphi_t^{Q^u}$ of the random vector $X_u$ under the $u$-survival risk-neutral measure $Q^u$. Its proof follows from \eqref{CF-1} and, being analogous to that of Lemma \ref{CF-P}, is omitted.

\begin{Lem}	\label{CF-Q}
For any $0\leq t \leq u \leq T$ and for any $z\in i\R^d$ the following holds:
\be
\varphi_t^{Q^u}\!(z) :=
E^{Q^u}\bigl[e^{\,z^{\top}X_u}|\,\F_t\bigr]
= \exp\,\Bigl(\Phi^Q(u-t,z)-\Phi^Q(u-t,0)+\bigl(\Psi^Q(u-t,z)-\Psi^Q(u-t,0)\bigr)^{\top}X_t\Bigr).
\ee
\end{Lem}

By combining Proposition \ref{gen} with Lemma \ref{surv-Q} and Lemma \ref{CF-Q} and using well-known Fourier inversion techniques, we can obtain semi-explicit formulas for a wide range of default-sensitive as well as equity/credit hybrid products. In particular, we now derive valuation formulas for \emph{Call} and \emph{Put} options (issued by a default-free third party) written on the defaultable stock $S$. We denote by $\Pi_{rf}(t,T):=E^Q\bigl[\exp(-\int_t^T\!\!r_sds)\,|\,\G_t\bigr]=E^Q\bigl[\exp(-\int_t^T\!\!r_sds)\,|\,\F_t\,\bigr]$ the arbitrage-free price at time $t\in[0,T]$ of a zero-coupon default-free bond. 

\begin{Cor}\label{CorOption}
For any $t\in[0,T]$ and for any strike price $K>0$ the following hold:
\begin{align}
C_K(t,T) &:= E^Q\Bigl[e^{-\int_t^T\!\!r_sds}(S_T-K)^+\bigr|\,\G_t\Bigr]	\nonumber\\
&= \frac{\Pi(t,T)}{2\pi}\int_{-\infty}^{+\infty}\!
\varphi_t^{Q^T}\!\!(0,\ldots,0,w+iu)\frac{K^{-(w-1+iu)}}{(w+iu)(w-1+iu)}\,du\,,
\label{Call}\\
P_K(t,T) &:= E^Q\Bigl[e^{-\int_t^T\!\!r_sds}(K-S_T)^+\bigr|\,\G_t\Bigr]	\nonumber\\
&= K\left(\Pi_{rf}(t,T)-\Pi(t,T)\right)
+\frac{\Pi(t,T)}{2\pi}\int_{-\infty}^{+\infty}\!
\varphi_t^{Q^T}\!\!(0,\ldots,0,y+iu)\frac{K^{-(y-1+iu)}}{(y+iu)(y-1+iu)}\,du\,.
\label{Put}
\end{align}
for some $w>1$ and $y<0$ such that the system of Riccati ODEs \eqref{Ric} has a unique solution for the initial conditions $z=(0,\ldots,0,w)^{\top}$ and $z=(0,\ldots,0,y)^{\top}$.
\end{Cor}
\begin{proof}
Observe first that: 
$$
E^Q\Bigl[e^{-\int_t^T\!\!r_sds}(S_T-K)^+\bigr|\,\G_t\Bigr]
= E^Q\Bigl[e^{-\int_t^T\!\!r_sds}\bigl(\tildeS_T-K\bigr)^+\,\ind_{\{\tau>T\}}\bigr|\,\G_t\Bigr]
= \Pi(t,T)E^{Q^T}\!\left[\bigl(e^{L_T}-K\bigr)^+\,\bigr|\,\F_t\right]
$$
where the second equality follows from \eqref{gen-1}. As in \cite{CM1} and \cite{Fi}, Lemma 10.2, it can be shown that:
$$
(e^x-K)^+ = \frac{1}{2\pi}\int_{-\infty}^{+\infty}\!e^{(w+iu)x}\frac{K^{-(w-1+iu)}}{(w+iu)(w-1+iu)}\,du
$$
for some $w>1$. Equation \eqref{Call} then follows by Fubini's theorem (see Cor. 2.5.21 of \cite{Fo2} for more details). Equation \eqref{Put} follows by an analogous computation once we observe that:
$$
E^Q\Bigl[e^{-\int_t^T\!\!r_sds}(K-S_T)^+\bigr|\,\G_t\Bigr]
= E^Q\Bigl[e^{-\int_t^T\!\!r_sds}\bigl(K-\tildeS_T\bigr)^+\,\ind_{\{\tau>T\}}\bigr|\,\G_t\Bigr]
+KE^Q\Bigl[e^{-\int_t^T\!\!r_sds}\left(1-\ind_{\{\tau>T\}}\right)\bigr|\,\G_t\Bigr].
$$
\end{proof}

If the discounted defaultable price process $\exp(-\int_0^{\cdot}\!r_udu)\,S$ is not only a $(Q,\GG)$-local martingale but also a true $(Q,\GG)$-martingale (this is for instance the case for the Heston with jump-to-default model considered in Section \ref{S6}; see \cite{Fo2}, Prop. 2.4.7), then the classical \emph{put-call parity} relation holds between the arbitrage-free prices of Call and Put options (issued by a default-free third party) with the same maturity $T$ and strike price $K$, written on the defaultable stock $S$, for all $t\in[0,T]$:
\be	\label{parity}
C_K(t,T) - P_K(t,T)
= E^Q\Bigl[e^{-\int_t^T\!\!r_sds}S_T\bigr|\,\G_t\Bigr] 
- KE^Q\Bigl[e^{-\int_t^T\!\!r_sds}\bigr|\,\G_t\Bigr]
= S_t - K\,\Pi_{rf}(t,T).
\ee
Note that, if the options are issued by an entity defaulting at $\tau$ (for instance, the defaultable firm itself), then the put-call parity relation \eqref{parity} still holds if the default-free bond $\Pi_{rf}(t,T)$ is replaced with the defaultable bond $\Pi(t,T)$.

\section{An example: the Heston with jump-to-default model}	\label{S6}

In this section, we illustrate some of the essential features of the proposed modeling framework within a simple example, which corresponds to a generalisation of the stochastic volatility model introduced by Heston \cite{He}, here extended by allowing the stock price process to be killed by a \emph{jump-to-default} event, in the spirit of \cite{CS}. 

\subsection{The model}	\label{S6-1}

Using the notation introduced in Section \ref{S2}, we let $d=3$ and $m=2$ and consider the following specification:
\be	\label{Hes-JTD}
A=
\begin{pmatrix} -k & 0 & 0	\\
0 & -k_0 & 0	\\
-1/2 & 0 & 0 \end{pmatrix}
\qquad
b = \begin{pmatrix} k\hat{v} \\ k_0\hat{y} \\ \mu \end{pmatrix}
\qquad
\Sigma=
\begin{pmatrix} \bar{\sigma} & 0 & 0	\\
0 & \sigma_0 & 0	\\
\rho & 0 & \sqrt{1-\rho^2}	\end{pmatrix}
\qquad
R_t = \begin{pmatrix} v_t & 0 & 0 \\ 0 & Y_t & 0 \\ 0 & 0 & v_t \end{pmatrix}
\ee
with $k\hat{v}\geq\bar{\sigma}^2/2$, $k_0\hat{y}\geq\sigma_0^2/2$ and $\rho\in\left[-1,1\right]$. The $P$-intensity $(\lambda^P_t)_{0\leq t \leq T}$ is specified as in equation \eqref{lambda-r}, i.e., we have $\lambda^P_t=\bar{\lambda}^P+\Lambda^P_1v_t+\Lambda^P_2Y_t$, for some $\bar{\lambda}^P,\Lambda^P_1,\Lambda^P_2\in\R_+$ with $\bar{\lambda}^P+\Lambda^P_1+\Lambda^P_2>0$. For simplicity, we assume that $r_t=\bar{r}\in\R_+$ for all $t\in[0,T]$. Note that this specification extends the Heston jump-to-default model considered in \cite{CS} by allowing $\lambda^P_t$ to depend on $v_t$ and on the additional stochastic factor $Y_t$. It can be easily checked that the specification \eqref{Hes-JTD} satisfies Assumption \ref{admissibility} and, due to Proposition \ref{SDE-S}, the defaultable stock price process $S=(S_t)_{0\leq t \leq T}$ has the following dynamics:
\be	\label{Hes-S}
dS_t = S_{t-}\,\bigl(\mu-\lambda^P_t\bigr)\,dt+S_{t-}\!\sqrt{v_t}\left(\rho\,dW^1_t+\sqrt{1-\rho^2}\,dW^3_t\right)-S_{t-}\,dM_t
\ee
where $M=(M_t)_{0\leq t \leq T}$ is the $(P,\GG)$-martingale defined by $M_t:=\ind_{\{\tau\leq t\}}-\int_0^{t\wedge\tau}\!\lambda^P_udu$. We also have:
\be	\label{Hes-v-X}	\begin{aligned}
dv_t &= k(\hat{v}-v_t)dt+\bar{\sigma}\sqrt{v_t}\,dW^1_t\,,	\\
dY_t &= k_0(\hat{y}-Y_t)dt+\sigma_0\sqrt{Y_t}\,dW^2_t\,.
\end{aligned} \ee

\subsection{Risk-neutral measures which preserve the Heston with jump-to-default structure}	\label{S6-2}

By relying on Theorem \ref{measure-change}, we now characterise the family of all risk-neutral measures $Q\in\Q$ which \emph{preserve the Heston with jump-to-default structure}, namely all risk-neutral measures $Q\in\Q$ which leave unchanged the structure of the SDEs \eqref{Hes-S}-\eqref{Hes-v-X} (compare also with \cite{Fo2}, Sect. 2.4.1).

\begin{Lem}	\label{mes-Heston}
A risk-neutral measure $Q\in\Q$ preserves the Heston with jump-to-default structure if and only if $dQ/dP$ admits the representation \eqref{RN} for some $\FF$-adapted processes $\theta=(\theta_t)_{0\leq t \leq T}$ and $\gamma=(\gamma_t)_{0\leq t \leq T}$ of the form \eqref{MPR-1} with $\bar{\theta}\in\R^3$ and $\Theta\in\R^{3\times 3}$ satisfying the following restrictions:
\be	\label{Hes-mes-change}
\hat{\theta} = 
\begin{pmatrix} 
\hat{\theta}_1 \\ 
\hat{\theta}_2 \\ 
\frac{\bar{r}+\bar{\lambda}^Q-\mu-\rho\,\hat{\theta}_1}{\sqrt{1-\rho^2}}
\end{pmatrix}
\qquad\qquad
\Theta =
\begin{pmatrix}
\Theta_{1,1} & 0 & 0	\\
0 & \Theta_{2,2} & 0	\\
\frac{\Lambda^Q_1-\rho\,\Theta_{1,1}}{\sqrt{1-\rho^2}} & \frac{\Lambda^Q_2}{\sqrt{1-\rho^2}} & 0
\end{pmatrix}
\ee
with $\hat{\theta}_1\geq\bar{\sigma}/2-k\hat{v}/\bar{\sigma}$ and $\hat{\theta}_2\geq\sigma_0/2-k_0\hat{y}/\sigma_0$.
\end{Lem}
\begin{proof}
The result follows from conditions \eqref{MPR-1}-\eqref{MPR-3} of Theorem \ref{measure-change}, noting that the preservation of the Heston with jump-to-default structure consists in the additional restriction $\Theta_{1,2}=\Theta_{2,1}=0$.
\end{proof}

\begin{Rem}
The parameter restrictions of Lemma \ref{mes-Heston} are significantly weaker than typical parameter restrictions found in the literature. For instance, let us consider the simpler default-free case (i.e., $\tau=+\infty$ $P$-a.s.) without the additional stochastic factor $Y$. In that case, the model \eqref{Hes-JTD}-\eqref{Hes-v-X} reduces to the classical (default-free) Heston \cite{He} stochastic volatility model. In their analysis of the existence of risk-neutral measures in stochastic volatility models, \cite{WH} show that there exists a risk-neutral measure $Q$ (preserving the Heston structure) if $\hat{\theta}_1=0$ and $\Theta_{1,1}\geq-k/\bar{\sigma}$ (see \cite{WH}, Thm. 3.5). In Lemma \ref{mes-Heston}, we show that such a risk-neutral measure exists without any restriction on $\Theta_{1,1}$ and also for non-trivial values of $\hat{\theta}_1$.
\end{Rem}

The main benefit of working with risk-neutral measures which preserve the Heston with jump-to-default structure consists in the possibility of obtaining closed-form solutions to the system of Riccati ODEs \eqref{Ric}, as shown in the next lemma (see also Remark \ref{grass}), which follows from Lemma 10.12 of \cite{Fi} by means of simple (but tedious and, hence, omitted) computations.

\begin{Lem}	\label{explicit-Ric}
Let $Q\in\Q$ be a risk-neutral measure which preserves the Heston with jump-to-default structure. Then the system of Riccati ODEs \eqref{Ric} admits the following solution, for all $z\in\C^2_-\times i\R$:
\begin{footnotesize}
$$	\begin{aligned}
\Psi_1^Q(t,z)
&= -\frac{\bigl(z_3-z_3^2+2\Lambda_1^Q(1-z_3)\bigr)\bigl(e^{\sqrt{\Delta_1}t}-1\bigr)
-\Bigl(\sqrt{\Delta_1}\bigl(e^{\sqrt{\Delta_1}t}+1\bigr)
+\bigl(\bar{\sigma}\,(\Theta_{1,1}+\rho z_3)-k\bigr)\bigl(e^{\sqrt{\Delta_1}t}-1\bigr)\Bigr)\,z_1}
{\sqrt{\Delta_1}\bigl(e^{\sqrt{\Delta_1}t}+1\bigr)
-\bigl(\bar{\sigma}\,(\Theta_{1,1}+\rho z_3)-k\bigr)\bigl(e^{\sqrt{\Delta_1}t}-1\bigr)
-\bar{\sigma}^2\bigl(e^{\sqrt{\Delta_1}t}-1\bigr)\,z_1}	\\
\Psi_2^Q(t,z)
&= -\frac{2\Lambda_2^Q(1-z_3)\bigl(e^{\sqrt{\Delta_2}t}-1\bigr)
-\Bigl(\sqrt{\Delta_2}\bigl(e^{\sqrt{\Delta_2}t}+1\bigr)
+(\sigma_0\,\Theta_{2,2}-k_0)\bigl(e^{\sqrt{\Delta_2}t}-1\bigr)\Bigr)\,z_2}
{\sqrt{\Delta_2}\bigl(e^{\sqrt{\Delta_2}t}+1\bigr)
-\bigl(\sigma_0\Theta_{2,2}-k_0\bigr)\bigl(e^{\sqrt{\Delta_2}t}-1\bigr)
-\sigma_0^2\bigl(e^{\sqrt{\Delta_2}t}-1\bigr)\,z_2}	\\
\Psi_3^Q(t,z) &= z_3	\\
\Phi^Q(t,z) &= 
\frac{2(k\hat{v}+\bar{\sigma}\hat{\theta}_1)}{\bar{\sigma}^2}
\log\left(\frac{2\sqrt{\Delta_1}\exp\,\Bigl(\frac{\sqrt{\Delta_1}-(\bar{\sigma}(\Theta_{1,1}+\rho z_3)-k)}{2}t\Bigr)}
{\sqrt{\Delta_1}(e^{\sqrt{\Delta_1}t}+1)-(\bar{\sigma}(\Theta_{1,1}+\rho z_3)-k)(e^{\sqrt{\Delta_1}t}-1)-\bar{\sigma}^2(e^{\sqrt{\Delta_1}t}-1)z_1}\right)	\\
& +\frac{2(k_0\hat{y}+\sigma_0\hat{\theta}_2)}{\sigma_0^2}
\log\left(\frac{2\sqrt{\Delta_2}\exp\,\Bigl(\frac{\sqrt{\Delta_2}-(\sigma_0\Theta_{2,2}-k_0)}{2}t\Bigr)}{\sqrt{\Delta_2}(e^{\sqrt{\Delta_2}t}+1)-(\sigma_0\Theta_{2,2}-k_0)(e^{\sqrt{\Delta_2}t}-1)-\sigma_0^2(e^{\sqrt{\Delta_2}t}-1)z_2}\right)	\\
&+\bigl(r+\bar{\lambda}^Q\bigr)(z_3-1)t
\end{aligned} $$
\end{footnotesize}
where:
$$
\Delta_1:=\bigl(\bar{\sigma}\,(\Theta_{1,1}+\rho z_3)-k\bigr)^2
+\bar{\sigma}^2\bigl(z_3-z_3^2+2\Lambda_1^Q(1-z_3)\bigr),
$$
$$
\Delta_2:=\bigl(\sigma_0\,\Theta_{2,2}-k_0\bigr)^2+2\,\sigma_0^2\,\Lambda_2^Q(1-z_3).
$$
\end{Lem}

By combining the above lemma with the results of Sections \ref{S4}-\ref{S5}, we can efficiently solve risk management problems and compute arbitrage-free prices of general default-sensitive payoffs.

\begin{Rem}	\label{grass}
In the context of the model \eqref{Hes-JTD}-\eqref{Hes-v-X}, it may seem simplistic to restrict the attention to the set of risk-neutral measures which preserve the Heston with jump-to-default structure, i.e., to the set of risk premia processes $\theta=(\theta_t)_{0\leq t \leq T}$ which satisfy the restriction $\Theta_{1,2}=\Theta_{2,1}=0$ (see the proof of Lemma \ref{mes-Heston}). 
However, due to Theorem 4.1 of \cite{GT}, the system of Riccati ODEs \eqref{Ric} for the model \eqref{Hes-JTD}-\eqref{Hes-v-X} admits an explicit solution \emph{if and only if} $\Theta_{1,2}=\Theta_{2,1}=0$. In other words, the set of risk-neutral measures which preserve the Heston with jump-to-default structure characterised in Lemma \ref{mes-Heston} coincides with the set of risk-neutral measures under which system \eqref{Ric} admits a closed-form solution, which is given in Lemma \ref{explicit-Ric}. Of course, by relying on Theorem \ref{measure-change}, we can relax the requirement of the preservation of the Heston with jump-to-default structure with the weaker requirement of the preservation of the affine structure of $(X,\tau)$ but, in that case, one has to rely on numerical techniques for solving the Riccati system \eqref{Ric}.
\end{Rem}

\subsection{Numerical results}	\label{S6-3}

This section reports the results of some numerical experiments for the Heston with jump-to-default model \eqref{Hes-JTD}-\eqref{Hes-v-X}. We adopt the following parameters' specification: $k=0.565$, $\hat{v}=0.07$, $\bar{\sigma}=0.281$, $k_0=0.325$, $\hat{y}=0.003$, $\sigma_0=0.036$, $\mu=0.1$, $\rho=-0.558$. These values have been obtained in \cite{CWu} by calibrating (via filtering and maximum likelihood techniques) an analogous stochastic volatility jump-to-default model to market quotes of equity options and CDS spreads on the Citigroup company (period: 5/2002 - 5/2006). The remaining parameters appearing in \eqref{Hes-mes-change} are specified as $\bar{r}=0$, $\Theta_{1,1}=\Theta_{2,2}=0.002$, $\hat{\theta}_1=\hat{\theta}_2=0.001$ and $\Lambda_1^P=\Lambda_2^P=\bar{\lambda}^P=0.1225$.

As a first application, we compute the distribution function of the defaultable stock price $S_T$ in the case of survival. More specifically, we consider the model \eqref{Hes-JTD}-\eqref{Hes-v-X} under the physical probability measure $P$ and, by relying on formula \eqref{distr-fct-1} together with Lemma \ref{explicit-Ric}, we compute the surface $(T,x)\mapsto P\left(S_T\leq x , \tau > T\right)$, for $T\in[0.5,3.0]$ and $x\in[0.7,1.3]$, for $S_0=1$. Note that, from the computational point of view, this is an easy task in our modeling framework, since it only requires a one-dimensional numerical integration.
As can be observed from Figure \ref{fig:figure1}, the shape of the distribution function strongly depends on the time horizon $T$, with a distinct behavior for small and large values of $x$, due to the combined effects of diffusive and jump-type risks. Figure \ref{fig:distrPQ} shows that the distribution function of the defaultable stock price can be quite different under the physical and a risk-neutral probability measure, even in the case where the overall default probability is kept at the same level (i.e., we have $P(\tau\leq T)=Q(\tau\leq T)=0.4$), thus accounting for risk-aversion and providing an evidence of the flexibility induced by the possibility of changing the default intensity from the physical to a risk-neutral probability measure (to this regard, compare also the discussion preceding Remark \ref{risk-premia}).

\begin{figure}[ht]
\begin{minipage}[b]{0.45\linewidth}
\centering
\includegraphics[height=2.25in,width=2.75in,angle=0]{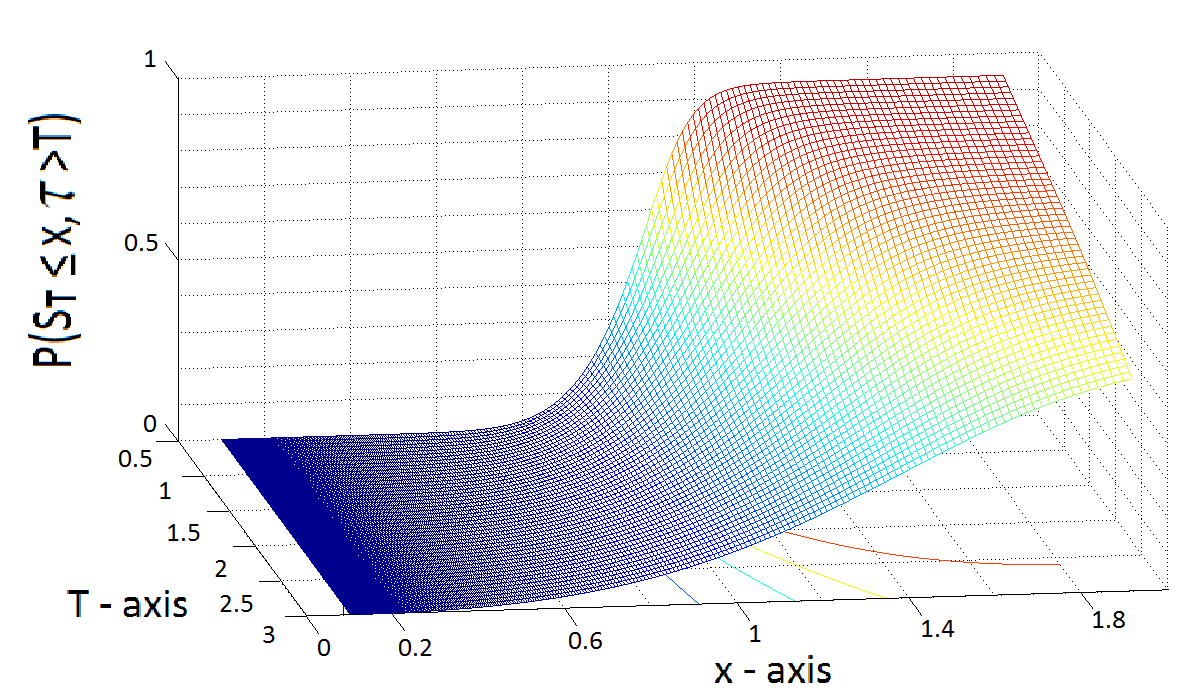}    
\caption{\footnotesize Surface $(T,x)\mapsto P\left(S_T\!\leq\!x,\tau\!>\!T\right)$ for the Heston with Jump-to-Default model.}
\label{fig:figure1}
\end{minipage}
\hspace{0.5cm}
\begin{minipage}[b]{0.45\linewidth}
\centering
\includegraphics[height=2in,width=2.75in,angle=0]{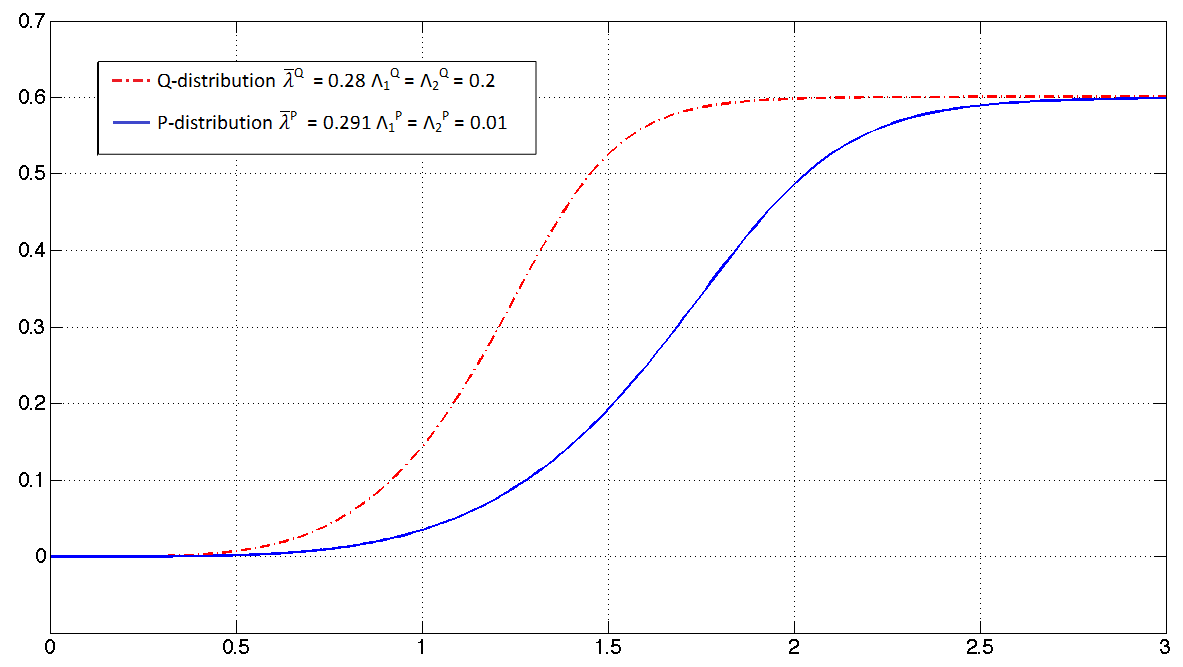}  
\caption{\footnotesize Distribution functions under physical ($P$) and risk-neutral ($Q$) probability measures.}
\label{fig:distrPQ}
\end{minipage}
\end{figure}

As a second application, we show the implied volatility surface generated by the model \eqref{Hes-JTD}-\eqref{Hes-v-X}. To this effect, we first compute a matrix of prices $P_K(0,T)$ of Put options on the defaultable stock $S_T$, issued by a default-free third party, with maturity $T\in[0.5,3.0]$ and moneyness $K/S_0\in[0.7,1.3]$, letting $\bar{\lambda}^Q=0.001$ and $\Lambda^Q_i=\Lambda^P_i$, for $i=1,2$. The computation is performed via the Fast Fourier Transform method of \cite{CM1}, by relying on Corollary \ref{CorOption} and Lemma \ref{explicit-Ric}. The corresponding implied volatilities are then computed by using the \textit{blsimpv} function in Matlab$^{\copyright}$ (R2012a 64-bit version).

\begin{figure}[ht]
\begin{minipage}[b]{0.45\linewidth}
\centering
\includegraphics[height=2.25in,width=2.75in,angle=0]{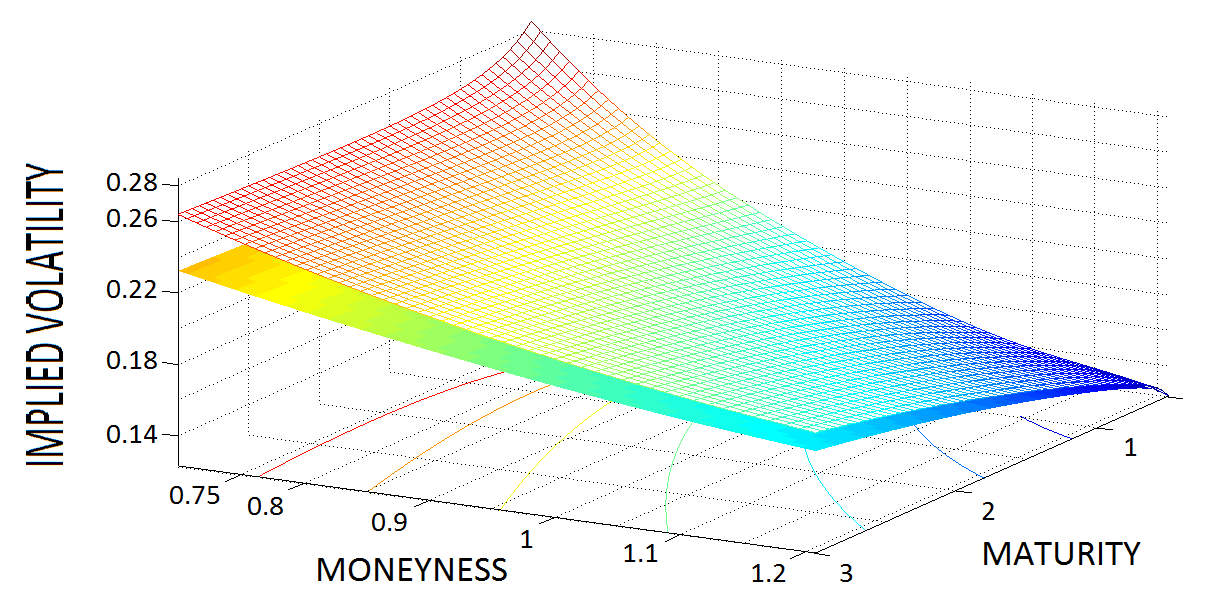}  
\caption{\footnotesize Implied volatility surfaces: standard Heston (solid) and Heston + Jump-to-Default (mesh).}
\label{fig:figure2}
\end{minipage}
\hspace{0.5cm}
\begin{minipage}[b]{0.45\linewidth}
\centering
\includegraphics[height=2in,width=2.8in,angle=0]{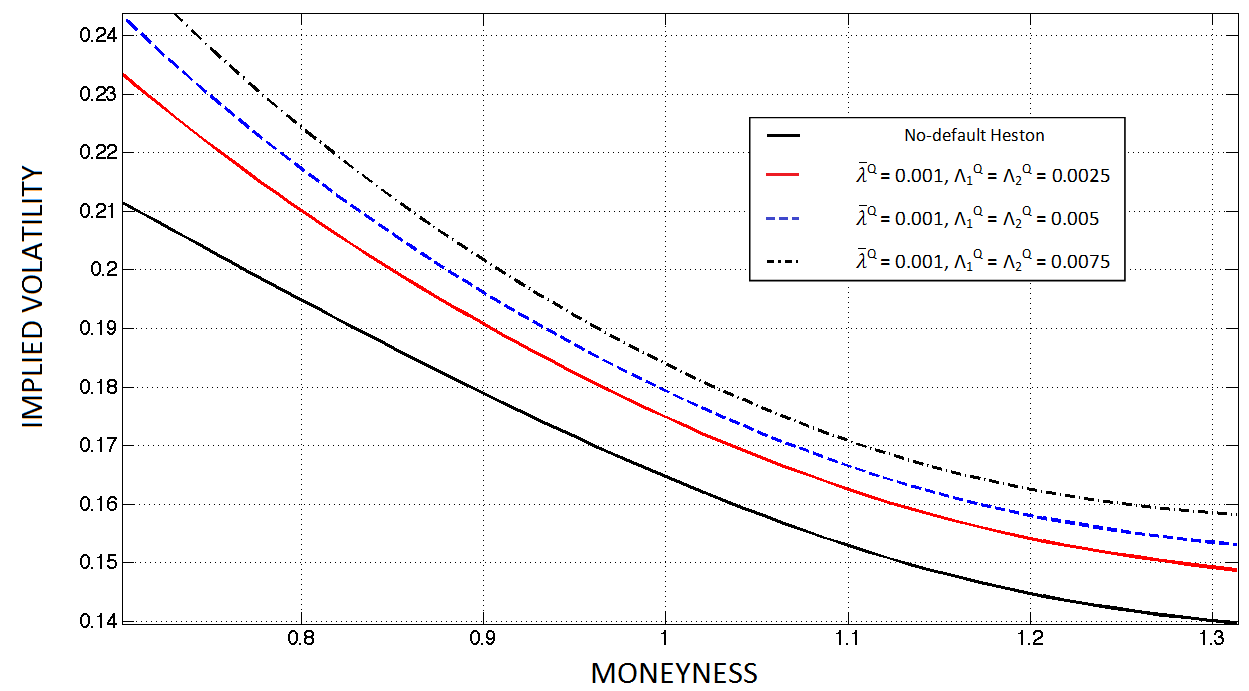}  
\caption{\footnotesize Implied volatility skew for different $Q$-intensities (for $T=1.75$).}
\label{fig:skew}
\end{minipage}
\end{figure}

Figure \ref{fig:figure2} compares the implied volatility surface generated by the model \eqref{Hes-JTD}-\eqref{Hes-v-X} with the implied volatility surface obtained from a standard (default-free) Heston \cite{He} model, i.e., by letting $\bar{\lambda}^Q\!=\!\Lambda^Q_1\!=\!\Lambda^Q_2\!=\!0$. It is evident that the introduction of default risk (through a jump-to-default) increases the implied volatility along all maturities and strikes. The increase is more pronounced for deep out-of-the-money options, due to the possibility of obtaining $K$ in the case of default (compare also with equation \eqref{Put}), thus confirming the fact that default risk is the main responsible for the value of out-of-the-money put options with short maturities. There is also a strong skew effect, which tends to flatten as the maturity increases but is always more significant than in the default-free case. 
The impact of default risk is also shown in Figure \ref{fig:skew}, which depicts the implied volatility skew for different specifications of the parameters which determine the default intensity $\lambda^Q$ together with the skew generated by a standard default-free Heston model, for the fixed maturity $T=1.75$. As expected, the implied volatility skew is more pronounced for a higher risk of default as measured by larger values of the default intensity parameters.

\section{Conclusions and further developments}	\label{S7}

We have proposed a general framework based on an affine process $X$ and on a doubly stochastic random time $\tau$ for the modeling of a defaultable stock. This approach allows to jointly model equity and credit risk, together with stochastic volatility and stochastic interest rate. Moreover, analytical tractability is ensured under both the physical and a set of risk-neutral probability measures, thanks to a flexible characterisation of all risk-neutral measures which preserve the affine structure of $(X,\tau)$. 
 
In the present paper, we have chosen to specify the driving process $X$ as an affine diffusion on $\R^m_{++}\!\times\R^{d-m}$, for some $m\in\{1,\ldots,d-1\}$. However, our techniques can be easily adapted to the more general case where $X$ is a continuous matrix-valued affine process (e.g., a Wishart process), as recently considered e.g. in \cite{CFMT}. We also want to mention that the characterisation of risk-neutral measures which preserve the affine structure of $(X,\tau)$ provided in Theorem \ref{measure-change} (or in Lemma \ref{mes-Heston} for the more specific case of the Heston with jump-to-default model) can also be useful in insurance mathematics for the valuation of mortality-indexed insurance contracts in the context of intensity-based mortality models (see e.g. \cite{Bif}).

\section*{Acknowledgements}
\vspace{-0.15cm}
\noindent 
The authors are thankful to Monique Jeanblanc and Wolfgang J. Runggaldier as well as to two anonymous referees for useful comments which helped to improve the paper. This research benefited from the support of the \emph{Chaire Risque de Crédit}, Fédération Bancaire Française.

\end{document}